\documentclass{article}

\usepackage{epsfig}
\usepackage{amssymb}
\usepackage{amsmath}
\usepackage{amsfonts}
\usepackage{amsthm}
\usepackage{a4wide}
\usepackage{hyperref}
\usepackage{tikz}
\usetikzlibrary{trees,arrows,positioning}

\newtheorem{theorem}{Theorem}
\newtheorem{corollary}{Corollary}
\newtheorem{algorithm}{Algorithm}

\author{Andreas Rosowski\\ University of Siegen, Germany\\ andreas.rosowski@student.uni-siegen.de}

\begin{document}

\title{Fast Commutative Matrix Algorithm}

\maketitle
\begin{abstract}
We show that the product of an $n \times 3$ matrix and a $3 \times 3$ matrix over a commutative ring can be computed using $6n+3$ multiplications. For two $3 \times 3$ matrices this gives us an algorithm using $21$ multiplications. This is an improvement with respect to Makarov's algorithm using $22$ multiplications\cite{makarov:matrmult}. We generalize our result for $n \times 3$ and $3 \times 3$ matrices and present an algorithm for computing the product of an $l \times n$ matrix and an $n \times m$ matrix over a commutative ring for odd $n$ using $n(lm + l + m -1)/2$ multiplications if $m$ is odd and using $(n(lm +l + m - 1) + l - 1)/2$ multiplications if $m$ is even. Waksman's and Islam's algorithm for odd $n$ needs $(n-1)(lm + l + m-1)/2 + lm$ multiplications \cite{islam, waksman}, thus in both cases less multiplications are required by our algorithm. We also give an algorithm for even $n$ using $n(lm + l + m -1)/2$ multiplications without making use of divisions, since Waksman's and Islam's algorithm make use of some divisions by $2$ \cite{islam, waksman}. Furthermore we present a novelty for matrix multiplication: In this paper we show that some non-bilinear algorithms with special properties can be used as recursive algorithms. In comparison to bilinear algorithms for small $n \times n$ matrices say $n<20$ we obtain some better results. From these non-bilinear algorithms we finally obtain approximate non-bilinear algorithms. For instance we obtain an approximate non-bilinear algorithm for $5 \times 5$ matrices that uses only $89$ multiplications. If at all it is possible to compare this algorithm with a bilinear algorithm we obtain a better result with respect to Smirnov's algorithm \cite{smirnov5x5}.
\end{abstract}

\section{Introduction}
In 1969 Strassen showed that the product of two $n \times n$ matrices can be computed using $O(n^{\log_2 (7)})$ arithmetic operations \cite{strassen:volker}. This work opened a new field of research and over the years better upper bounds for the exponent of matrix multiplication were published. In 1990 Coppersmith and Winograd obtained an upper bound of $2.375477$ for the exponent \cite{coppersmith:winograd}. For a long time this was the best result. Since 2010 further improvements were obtained in a series of papers \cite{stothers_d, le:gall, stothers, williams, williams_2}. The best result so far was published in 2014 by Le Gall who obtained an upper bound of $2.3728639$ for the exponent \cite{le:gall}.\\

In this paper we first study the product of an $n \times 3$ matrix $A$ and a $3 \times 3$ matrix $B$ over a commutative ring and show that we can compute the product $AB$ using $6n + 3 $ multiplications. The basic idea is to improve the computation of the product of a $1 \times 3$ vector $a$ and a $3 \times 3$ matrix $B$ over a commutative ring in the sense that we try to obtain as much as possible multiplications that contain only entries of the matrix $B$ but without using more than $9$ multiplications overall. The multiplications which contain only entries of the matrix $B$ only need to be calculated once and can therefore be reused in the matrix multiplication. In the special case $n=3$ we obtain an algorithm for multiplying two $3 \times 3$ matrices using $21$ multiplications which improves the best result so far from Makarov using $22$ multiplications \cite{makarov:matrmult}. Our next step is to generalize this result to the computation of the product of an $l \times n$ matrix $A$ and an $n \times m$ matrix $B$ over a commutative ring for odd $n$. We show that the product $AB$ can be computed using $n(lm + l + m - 1)/2$ multiplications if $m$ is odd and $(n(lm +l + m - 1) + l - 1)/2$ multiplications if $m$ is even. This improves Waksman's and Islam's algorithm which require $(n-1)(lm + l + m - 1)/2 + lm$ multiplications for odd $n$ \cite{islam, waksman}. Additionally we present an algorithm to multiply an $l \times n$ matrix and an $n \times m$ matrix using $n(lm + l + m - 1)/2$ multiplications without making use of divisions, whereas Waksman's and Islam's algorithm makes use of some divisions by $2$ \cite{islam, waksman}.\\

Finally we present a novelty for matrix multiplication. We can show that a non-bilinear algorithm for matrix multiplication can be applied recursively under special circumstances. For this we first generalize the term of bilinear algorithms. Bilinear algorithms have the property that every multiplication is in bilinear form. For the computation of a matrix product $AB$ this means that in every product that is computed the left (resp., right) factor is
a sum of entries from $A$ (resp., $B$). The consequence is that the linear combinations of the products are also in bilinear form. We have found algorithms whose multiplications are not in bilinear form but the linear combinations of the products are in bilinear form. Such an algorithm can be seen as a generalization of bilinear algorithms if one defines a bilinear algorithm independent from the form of the multiplications as an algorithm whose linear combinations are in bilinear form. But in order to distinguish such algorithms from usual bilinear algorithms we call such an algorithm non-bilinear. Of course a non-bilinear algorithm can be applied recursively. Note that by this definition every bilinear algorithm is also a non-bilinear algorithm. We have found formulas for the computation of the product of two $n \times n$ matrices. For even $n$ the algorithm needs $n(n^2+3n+1)/2$ multiplications and $n(n^2+3n+2)/2$ multiplications are required for odd $n$. In the case $n=14$ this gives an algorithm using $1673$ multiplications and an exponent of roughly $2.8125$.

Additionally we present approximate non-bilinear algorithm. Approximate algorithms for matrix multiplication were introduced for the first time by Bini et al. \cite{bini}. We give a general formula for the computation of the product of two $n \times n$ matrices. For even $n$ the algorithm needs $n(n^2+2n)/2$ multiplications and $n(n^2+2n+1)/2$ multiplications are required for odd $n$. In the case $n=5$ we were even able to find an approximate non-bilinear algorithm that uses only $89$ multiplications.

\subsection{Related Work}
In this section we present some related work. We start with presenting some results about multiplication of two square matrices. Since Strassen showed in 1969 that the product of two matrices can be computed using $O(n^{\log_2 (7)})$ arithmetic operations \cite{strassen:volker} and since it is shown that for $2 \times 2$ matrices $7$ is the optimal number of multiplications \cite{hopcroft:kerr:min, winograd}, it is interesting to study $n \times n$ matrices for $n \geq 3$, to obtain an even faster algorithm for matrix multiplication. For $3 \times 3$ matrices $21$ multiplications would be needed to obtain an even faster algorithm than Strassen's since $\log_3 (21) \approx 2.7712 < 2.807 \approx \log_2 (7)$. In 1976 Laderman obtained a non-commutative $3 \times 3$ algorithm using only $23$ multiplications \cite{lader:man}. It is not known if there exists a non-commutative $3 \times 3$ algorithm that uses $22$ or less multiplications.

Hopcroft and Musinski showed in \cite{hopcroft:musinski} that the number of multiplications of non-commutative algorithms to compute the product of an $l \times n$ matrix and an $n \times m$ matrix is the same number that is required to compute the product of an $n \times m$ matrix and an $m \times l$ matrix and of an $l \times m$ matrix and an $m \times n$ matrix etc. This means if one finds an algorithm for the product of an $l \times m$ matrix and an $m \times n$ matrix using $x$ multiplications there exists a matrix multiplication algorithm for $lnm \times lnm$ matrices using $x^3$ multiplications overall. This algorithm for square matrices will then have an exponent of $\log_{lnm} (x^3)$.

We present some examples of non-square matrix multiplication algorithms. In \cite{hopcroft:kerr} Hopcroft and Kerr showed that the product of a $p \times 2$ matrix and a $2 \times n$ matrix can be multiplied using $\lceil (3pn + max\{n, p\})/2 \rceil$ multiplications without using commutativity. In the case $p=3=n$ this gives an algorithm using $15$ multiplications. Combined with the results of \cite{hopcroft:musinski} this gives an algorithm for $18 \times 18$ matrices using $15^3 = 3375$ multiplications and an exponent of $\log_{18} (3375) \approx 2.811$. Smirnov obtained an algorithm for the product of a $3 \times 3$ matrix and a $3 \times 6$ matrix using $40$ multiplications\cite{smirnov}. By \cite{hopcroft:musinski} this gives an algorithm for $54 \times 54$ matrices using $40^3 = 64000$ multiplications and an exponent of $\log_{54} (64000) \approx 2.7743$.\\

Cariow et al. developed a high-speed parallel $3 \times 3$ matrix multiplier structure based on the commutative $3 \times 3$ matrix algorithm using $22$ multiplications obtained by Makarov \cite{hardware, makarov:matrmult}. We suppose that the structure could be improved by using our commutative $3 \times 3$ matrix algorithm using $21$ multiplications.

In \cite{optimization:small:matrices} Drevet et al. optimized the number of required multiplications of small matrices up to $30 \times 30$ matrices. They considered non-commutative and commutative algorithms. Combined with our results for commutative rings we suppose that some results could be improved.

\section{Matrix Multiplication over a Commutative Ring}
Let $R$ denote a commutative ring throughout this Section.
\subsection{Product of $n \times 3$ and $3 \times 3$ Matrices}
Consider the product of an $1 \times 3$ vector $a$ and a $3 \times 3$ matrix $B$ over a commutative ring.
\begin{equation}\label{foralgorithm1}
a = \begin{bmatrix}
a_{1} & a_{2} & a_{3}
\end{bmatrix}
\textbf{~~~~~~~~~~}
B = \begin{bmatrix}
b_{11} & b_{12} & b_{13}\\
b_{21} & b_{22} & b_{23}\\
b_{31} & b_{32} & b_{33}
\end{bmatrix}
\end{equation}
In the usual way the product of $a$ and $B$ would be computed as:
\begin{displaymath}
aB = \begin{bmatrix}
a_1b_{11} + a_2b_{21} + a_3b_{31} &~& a_1b_{12} + a_2b_{22} + a_3b_{32} &~& a_1b_{13} + a_2b_{23} + a_3b_{33}
\end{bmatrix}
\end{displaymath}
But it can also be computed by first computing these $9$ products:
\begin{algorithm}\label{matrix3x3}
Input: Vector $a$ and Matrix $B$ as in \eqref{foralgorithm1}.\\
Let
\begin{align*}
p_1&:=(a_2 + b_{12})(a_1 + b_{21})\\
p_2&:=(a_3 + b_{13})(a_1 + b_{31})\\
p_3&:=(a_3 + b_{23})(a_2 + b_{32})\\
p_4&:=a_1(b_{11} - b_{12} - b_{13} - a_2 - a_3)\\
p_5&:=a_2(b_{22} - b_{21} - b_{23} - a_1 - a_3)\\
p_6&:=a_3(b_{33} - b_{31} - b_{32} - a_1 - a_2)\\
p_7&:=b_{12}b_{21}\\
p_8&:=b_{13}b_{31}\\
p_9&:=b_{23}b_{32}
\end{align*}
Output:
\begin{displaymath}
aB = \begin{bmatrix}
p_4 + p_1 + p_2 - p_7 - p_8 &~& p_5 + p_1 + p_3 - p_7 - p_9 &~& p_6 + p_2 + p_3 - p_8 - p_9
\end{bmatrix}
\end{displaymath}
\end{algorithm}
\begin{theorem}
Let $R$ be a commutative ring, let $n \geq 1$, let $A$ be an $n \times 3$ matrix over $R$ and let $B$ be a $3 \times 3$ matrix over $R$. Then the product $AB$ can be computed using $6n+3$ multiplications.
\end{theorem}
\begin{proof}
Consider Algorithm \ref{matrix3x3}. The products $p_7$, $p_8$ and $p_9$ contain only entries of the matrix $B$. One can observe that for all $n \geq 1$ the products $p_7$, $p_8$ and $p_9$ can be reused for the product $AB$ and therefore $3(n-1)$ multiplications are saved.
\end{proof}
We give an example. In the case $n=3$ we obtain an algorithm with $21$ multiplications for the matrix multiplication. This algorithm needs one multiplication less than Makarov's \cite{makarov:matrmult}.
\begin{corollary}
Let $R$ be a commutative ring and let $A$ and $B$ be $3 \times 3$ matrices over $R$ as shown below. Then the product $AB$ can be computed using $21$ multiplications as follows:
\begin{align*}
A = \begin{bmatrix}
a_{11} & a_{12} & a_{13}\\
a_{21} & a_{22} & a_{23}\\
a_{31} & a_{32} & a_{33}
\end{bmatrix}
&&
B = \begin{bmatrix}
b_{11} & b_{12} & b_{13}\\
b_{21} & b_{22} & b_{23}\\
b_{31} & b_{32} & b_{33}
\end{bmatrix}
\end{align*}
\begin{align*}
p_1&:=(a_{12} + b_{12})(a_{11} + b_{21})\\
p_2&:=(a_{13} + b_{13})(a_{11} + b_{31})\\
p_3&:=(a_{13} + b_{23})(a_{12} + b_{32})\\
p_4&:=a_{11}(b_{11} - b_{12} - b_{13} - a_{12} - a_{13})\\
p_5&:=a_{12}(b_{22} - b_{21} - b_{23} - a_{11} - a_{13})\\
p_6&:=a_{13}(b_{33} - b_{31} - b_{32} - a_{11} - a_{12})
\end{align*}
\begin{align*}
p_7&:=(a_{22} + b_{12})(a_{21} + b_{21})\\
p_8&:=(a_{23} + b_{13})(a_{21} + b_{31})\\
p_9&:=(a_{23} + b_{23})(a_{22} + b_{32})\\
p_{10}&:=a_{21}(b_{11} - b_{12} - b_{13} - a_{22} - a_{23})\\
p_{11}&:=a_{22}(b_{22} - b_{21} - b_{23} - a_{21} - a_{23})\\
p_{12}&:=a_{23}(b_{33} - b_{31} - b_{32} - a_{21} - a_{22})
\end{align*}
\begin{align*}
p_{13}&:=(a_{32} + b_{12})(a_{31} + b_{21})\\
p_{14}&:=(a_{33} + b_{13})(a_{31} + b_{31})\\
p_{15}&:=(a_{33} + b_{23})(a_{32} + b_{32})\\
p_{16}&:=a_{31}(b_{11} - b_{12} - b_{13} - a_{32} - a_{33})\\
p_{17}&:=a_{32}(b_{22} - b_{21} - b_{23} - a_{31} - a_{33})\\
p_{18}&:=a_{33}(b_{33} - b_{31} - b_{32} - a_{31} - a_{32})\\\\
p_{19}&:=b_{12}b_{21}\\
p_{20}&:=b_{13}b_{31}\\
p_{21}&:=b_{23}b_{32}
\end{align*}
Hence,
\begin{displaymath}
AB = \begin{bmatrix}
p_4 + p_1 + p_2 - p_{19} - p_{20} &~& p_5 + p_1 + p_3 - p_{19} - p_{21} &~& p_6 + p_2 + p_3 - p_{20} - p_{21}\\
p_{10} + p_7 + p_8 - p_{19} - p_{20} &~& p_{11} + p_7 + p_9 - p_{19} - p_{21} &~& p_{12} + p_8 + p_9 - p_{20} - p_{21}\\
p_{16} + p_{13} + p_{14} - p_{19} - p_{20} &~& p_{17} + p_{13} + p_{15} - p_{19} - p_{21} &~& p_{18} + p_{14} + p_{15} - p_{20} - p_{21}
\end{bmatrix}
\end{displaymath}
\end{corollary}

\subsection{General Matrix Multiplication}
We start this Section with presenting an algorithm for the product of $l \times n$ and $n \times m$ matrices over a commutative ring for even $n$ using $n(lm + l + m - 1)/2$ multiplications. In the literature there is already an algorithm for such matrices with $n(lm + l + m - 1)/2$ multiplications \cite{islam, waksman}, but this algorithm makes use of some divisions by $2$. In this Section we present an algorithm that does not use divisions.
\begin{theorem}\label{commutativ_even_n}
Let $R$ be a commutative ring, let $n \geq 2$ be even and let $l,m \geq 1$. Furthermore let $A \in R^{l \times n}$, $B \in R^{n \times m}$ be matrices. Then the product $AB$ can be computed using $n(lm + l + m - 1)/2$ multiplications without making use of divisions.
\end{theorem}
\begin{proof}
Let $A$ and $B$ be matrices as in the Theorem. Let $a_{ij}$ denote the entries of $A$ and let $b_{ij}$ denote the entries of $B$ and let $c_{ij}$ denote the entries of $AB$. The product $AB$ can be computed as follows. For $i = 1,\dots,l$ let
\begin{displaymath}
c_{i1} = \sum_{k=1}^{n/2} a_{i(2k-1)}(b_{(2k-1)1} + a_{i(2k)}) + \sum_{k=1}^{n/2} a_{i(2k)}(b_{(2k)1} - a_{i(2k-1)})
\end{displaymath}
and for $i = 1,\dots,l$ and $j=2,\dots,m$ let
\begin{displaymath}
c_{ij} = \sum_{k=1}^{n/2} (a_{i(2k-1)} + b_{(2k)j})(a_{i(2k)} + b_{(2k-1)1} + b_{(2k-1)j}) - \sum_{k=1}^{n/2} a_{i(2k-1)}(b_{(2k-1)1} + a_{i(2k)})
\end{displaymath}
\begin{displaymath}
-\sum_{k=1}^{n/2} b_{(2k)j}(b_{(2k-1)1} + b_{(2k-1)j})
\end{displaymath}
These formulas can easily be verified. It remains to show that only $n(lm + l + m - 1)/2$ multiplications are required by these formulas. To show the number of multiplications we count the multiplications. Since the index $i$ takes $l$ different values, the index $j$ takes $m-1$ different values and the index $k$ takes $n/2$ different values there are $ln/2$ multiplications of the form $a_{i(2k-1)}(b_{(2k-1)1} + a_{i(2k)})$, $ln/2$ multiplications of the form $a_{i(2k)}(b_{(2k)1} - a_{i(2k-1)})$, $n(m-1)/2$ multiplications of the form $b_{(2k)j}(b_{(2k-1)1} + b_{(2k-1)j})$ and $ln(m-1)/2$ multiplications of the form $(a_{i(2k-1)} + b_{(2k)j})(a_{i(2k)} + b_{(2k-1)1} + b_{(2k-1)j})$. Thus, overall $ln/2 + ln/2 + n(m-1)/2 + ln(m-1)/2 = n(lm + l + m - 1)/2$ multiplications are required.
\end{proof}

Algorithm \ref{matrix3x3} from Section $2.1$ is the basic idea of a general algorithm for the matrix product of $l \times n$ and $n \times m$ matrices over a commutative ring for odd $n$. The algorithm we present below is split into two cases. In Case $1$ $m$ is odd and in Case $2$ $m$ is even. This leads us to the following:
\begin{theorem}
Let $R$ be a commutative ring, let $n \geq 3$ be odd, $l \geq 1$, $m \geq 3$ and let $A \in R^{l \times n}$, $B \in R^{n \times m}$ be matrices. Then the following holds:
\begin{itemize}
\item If $m$ is odd the product $AB$ can be computed using $n(lm + l + m - 1)/2$ multiplications.
\item If $m$ is even the product $AB$ can be computed using $(n(lm +l + m - 1) + l - 1)/2$ multiplications.
\end{itemize}
\end{theorem}
\begin{proof}
Let $A$ and $B$ be matrices as in the Theorem. Now split $A$ and $B$ in submatrices in the following way:\\\\
$A = \begin{bmatrix}
A_1 & A_2
\end{bmatrix}$,
with $A_1 \in R^{l \times 3}$ and $A_2 \in R^{l \times n-3}$,\\\\
$B = \begin{bmatrix}
B_1\\
B_2
\end{bmatrix}$,
with $B_1 \in R^{3 \times m}$ and $B_2 \in R^{n-3 \times m}$.

Then $AB = A_1B_1 + A_2B_2$. By using Theorem \ref{commutativ_even_n} the product $A_2B_2$ can be computed using $(n-3)(lm + l + m -1)/2$ multiplications. Let $a_{ij}$ denote the entries of $A_1$ and let $b_{ij}$ denote the entries of $B_1$ and let $c_{ij}$ denote the entries of $A_1B_1$. The matrix $A_1B_1$ can be computed as follows.\\\\
\underline{Case 1}: $m$ is odd.\\
For $i=1,\dots,l$ let\\\\
$c_{i1} = (a_{i1} + b_{21})(a_{i2} + b_{12}) + (a_{i1} + b_{31})(a_{i3} + b_{13}) + a_{i1}(b_{11} - b_{12} - b_{13} - a_{i2} - a_{i3}) - b_{12}b_{21} - b_{13}b_{31}$\\\\
$c_{i2} = (a_{i1} + b_{21})(a_{i2} + b_{12}) + (a_{i2} + b_{32})(a_{i3} + b_{23}) + a_{i2}(b_{22} - b_{21} - b_{23} - a_{i1} - a_{i3}) - b_{12}b_{21} - b_{23}b_{32}$\\\\
$c_{i3} = (a_{i1} + b_{31})(a_{i3} + b_{13}) + (a_{i2} + b_{32})(a_{i3} + b_{23}) + a_{i3}(b_{33} - b_{31} - b_{32} - a_{i1} - a_{i2}) - b_{13}b_{31} - b_{23}b_{32}$\\\\
and for $i=1,\dots,l$ and $j=4,6,8,\dots,m-1$ let\\\\
$c_{ij} = (a_{i1} + b_{21})(a_{i2} + b_{12}) + (a_{i1} + b_{31})(a_{i3} + b_{13}) + (a_{i1} + b_{21} - b_{2j})(-a_{i2} - b_{12} + b_{1j} - b_{1(j+1)})$\\
$+ (a_{i1} + b_{31} - b_{3j})(-a_{i3} - b_{13} + b_{1(j+1)}) - b_{12}b_{21} - b_{13}b_{31} - (b_{21} - b_{2j})(-b_{12} + b_{1j} - b_{1(j+1)})$\\
$- (b_{31} - b_{3j})(-b_{13} + b_{1(j+1)})$\\\\
$c_{i(j+1)} = (a_{i1} + b_{31})(a_{i3} + b_{13}) + (a_{i2} + b_{32})(a_{i3} + b_{23}) + (a_{i1} + b_{31} - b_{3j})(-a_{i3} - b_{13} + b_{1(j+1)})$\\
$+ (a_{i2} + b_{32} + b_{3j} - b_{3(j+1)})(-a_{i3} - b_{23} + b_{2(j+1)}) - b_{13}b_{31} - b_{23}b_{32} - (b_{31} - b_{3j})(-b_{13} + b_{1(j+1)})$\\
$- (b_{32} + b_{3j} - b_{3(j+1)})(-b_{23} + b_{2(j+1)})$\\

It can easily be seen that $6l + 3 + 3l(m-3)/2 + 3(m-3)/2 = 3(lm + l + m - 1)/2$ multiplications are required to compute $A_1B_1$.

Thus, $AB$ can be computed using $3(lm + l + m - 1)/2 + (n-3)(lm + l + m -1)/2 = n(lm + l + m -1)/2$ multiplications.\\\\
\underline{Case 2}: $m$ is even.\\
For $i=1,\dots,l$ let\\\\
$c_{i1} = (a_{i1} + b_{21})(a_{i2} + b_{12}) + (a_{i1} + b_{31})(a_{i3} + b_{13}) + a_{i1}(b_{11} - b_{12} - b_{13} - a_{i2} - a_{i3}) - b_{12}b_{21} - b_{13}b_{31}$\\\\
$c_{i2} = (a_{i1} + b_{21})(a_{i2} + b_{12}) + (a_{i2} + b_{32})(a_{i3} + b_{23}) + a_{i2}(b_{22} - b_{21} - b_{23} - a_{i1} - a_{i3}) - b_{12}b_{21} - b_{23}b_{32}$\\\\
$c_{i3} = (a_{i1} + b_{31})(a_{i3} + b_{13}) + (a_{i2} + b_{32})(a_{i3} + b_{23}) + a_{i3}(b_{33} - b_{31} - b_{32} - a_{i1} - a_{i2}) - b_{13}b_{31} - b_{23}b_{32}$\\\\
$c_{i4} = (a_{i1} + b_{21})(a_{i2} + b_{12}) + (a_{i1} + b_{21} - b_{24})(-a_{i2} - b_{12} + b_{14}) + a_{i3}b_{34} - b_{12}b_{21} - (b_{21} - b_{24})(-b_{12} + b_{14})$\\\\
and for $i=1,\dots,l$ and $j=5,7,9,\dots,m-1$ let\\\\
$c_{ij} = (a_{i1} + b_{21})(a_{i2} + b_{12}) + (a_{i1} + b_{31})(a_{i3} + b_{13}) + (a_{i1} + b_{21} - b_{2j})(-a_{i2} - b_{12} + b_{1j} - b_{1(j+1)})$\\
$+ (a_{i1} + b_{31} - b_{3j})(-a_{i3} - b_{13} + b_{1(j+1)}) - b_{12}b_{21} - b_{13}b_{31} - (b_{21} - b_{2j})(-b_{12} + b_{1j} - b_{1(j+1)})$\\
$- (b_{31} - b_{3j})(-b_{13} + b_{1(j+1)})$\\\\
$c_{i(j+1)} = (a_{i1} + b_{31})(a_{i3} + b_{13}) + (a_{i2} + b_{32})(a_{i3} + b_{23}) + (a_{i1} + b_{31} - b_{3j})(-a_{i3} - b_{13} + b_{1(j+1)})$\\
$+ (a_{i2} + b_{32} + b_{3j} - b_{3(j+1)})(-a_{i3} - b_{23} + b_{2(j+1)}) - b_{13}b_{31} - b_{23}b_{32} - (b_{31} - b_{3j})(-b_{13} + b_{1(j+1)})$\\
$- (b_{32} + b_{3j} - b_{3(j+1)})(-b_{23} + b_{2(j+1)})$\\

One can easily verify that in this case $8l + 4 + 3l(m-4)/2 + 3(m-4)/2 = 2(l - 1) + 3(lm + m)/2$ multiplications are required to compute $A_1B_1$.

Thus, $AB$ can be computed using $2(l - 1) + 3(lm + m)/2 + (n-3)(lm + l + m -1)/2 = (n(lm +l + m - 1) + l - 1)/2$ multiplications.
\end{proof}
In both cases less multiplications are required to compute $AB$ than Waksman's and Islam's algorithm \cite{islam, waksman} for odd $n$ requires.

\section{Non-Bilinear Algorithms}
Surprisingly it is indeed possible to use a non-bilinear algorithm for matrix multiplication as a recursive algorithm. The crucial point is the form of the results of the linear combinations of the multiplications. The results need to be in bilinear form. Another point is that the identity $AB = BA$ for two $n \times n$ matrices $A$ and $B$ in general does not hold. But for matrices over a commutative ring at least $AB = (B^TA^T)^T$ holds and we will use this fact for our next algorithms.

In this Section we will consider an arbitrary (not necessarily commutative) ring $R$ on which we define an involution $f:R \rightarrow R$ defined by $A \mapsto A^T$ with the following axioms:
\begin{equation}\label{involution}
(AB)^T = B^TA^T
\end{equation}
\begin{displaymath}
(A+B)^T = A^T + B^T
\end{displaymath}

In the sense of matrix multiplication $A^T$ denotes the transpose of a matrix $A$. Note that the following algorithms can be applied recursively if we choose $R$ to be the ring of all square matrices over a commutative ring. Then the entries $A_{ij}, B_{ij}$ and $C_{ij}$ are submatrices.

\begin{theorem}\label{commutative_recursive_1}
Let $R$ be a ring satisfying \eqref{involution}, let $n \geq 2$ be even. Furthermore let $A$ and $B$ be $n \times n$ matrices over $R$. Then there is a non-bilinear algorithm using $n(n^2+3n+1)/2$ multiplications for computing the product $AB$.
\end{theorem}
\begin{proof}
Let $A$ and $B$ be matrices as in the Theorem. Let $A_{ij}$ denote the entries of $A$ and let $B_{ij}$ denote the entries of $B$ and let $C_{ij}$ denote the entries of $AB$. The product $AB$ can be computed as follows: For $i = 1,3,5,\dots,n-1$ and $j=1,\dots,n$ let
\begin{align*}
C_{ij}&= \sum_{k=1}^n (A_{ik} + B_{jk}^T)(A_{(i+1)j}^T + B_{kj}) - (\sum_{k=1}^n A_{ik} + \sum_{k=1}^n B_{jk}^T)A_{(i+1)j}^T - \sum_{k=1}^n B_{jk}^TB_{kj}\\
C_{(i+1)j}&= \sum_{k=1}^n (A_{(i+1)k} + B_{jk}^T)(A_{ij}^T + B_{kj}) - (\sum_{k=1}^n A_{(i+1)k} + \sum_{k=1}^n B_{jk}^T)A_{ij}^T - \sum_{k=1}^n B_{jk}^TB_{kj}
\end{align*}

The correctness can easily be verified. It remains to show that only $n(n^2+3n+1)/2$ multiplications are required by these formulas. To show this we first show that for every $l \in \{1,\dots,n\}$ every multiplication of the form $(A_{il} + B_{jl}^T)(A_{(i+1)j}^T + B_{lj})$ is used for the computation of $C_{ij}$ and $C_{(i+1)l}$. It is clear by definition that every multiplication of the form $(A_{il} + B_{jl}^T)(A_{(i+1)j}^T + B_{lj})$ is used for the computation of $C_{ij}$. We have
\begin{displaymath}
((A_{il} + B_{jl}^T)(A_{(i+1)j}^T + B_{lj}))^T = (A_{(i+1)j} + B_{lj}^T)(A_{il}^T + B_{jl})
\end{displaymath}

Consider the first term in the computation of $C_{(i+1)l}$:
\begin{displaymath}
\sum_{k=1}^n (A_{(i+1)k} + B_{lk}^T)(A_{il}^T + B_{kl})
\end{displaymath}

Since the index $k$ runs over $1,\dots,n$ and $j \in \{1,\dots,n\}$ at some point the index $k$ takes the value $j$ and from this it follows that the multiplications of the form $(A_{il} + B_{jl}^T)(A_{(i+1)j}^T + B_{lj})$ are also used in the computation of $C_{(i+1)l}$. The next step is to count the multiplications. The index $i$ takes $n/2$ different values. The index $j$ takes $n$ different values. Thus overall $n^2$ multiplications of the form $(\sum_{k=1}^n A_{ik} + \sum_{k=1}^n B_{jk}^T)A_{(i+1)j}^T$ and $(\sum_{k=1}^n A_{(i+1)k} + \sum_{k=1}^n B_{jk}^T)A_{ij}^T$ are computed. Obviously every multiplication of the form $B_{jk}^TB_{kj}$ is used for the computation of $C_{ij}$ and $C_{(i+1)j}$. But $B_{jk}^TB_{kj}$ is also used for the computation of $C_{ik}$ and $C_{(i+1)k}$. Since the multiplications of the form $B_{jk}^TB_{kj}$ are independent from the index $i$ for every $i$ they can be reused. Overall we obtain $n(n-1)/2$ multiplications of the form $B_{jk}^TB_{kj}$ for $k \neq j$ and $n$ multiplications of the form $B_{jk}^TB_{kj}$ for $k = j$ that are computed. Since we have shown that every multiplication of the form $(A_{ik} + B_{jk}^T)(A_{(i+1)j}^T + B_{kj})$ is used for the computation of $C_{ij}$ and $C_{(i+1)k}$ it suffices to count only the multiplications of the form $(A_{(i+1)k} + B_{jk}^T)(A_{ij}^T + B_{kj})$ that are used for the computation of $C_{(i+1)j}$. The index $i$ takes $n/2$ different values, the index $j$ takes $n$ different values and the index $k$ runs over $1,\dots,n$. Hence $n \cdot n \cdot n/2 = n^3/2$ multiplications of the form $(A_{(i+1)k} + B_{jk}^T)(A_{ij}^T + B_{kj})$ are computed. By this the total number of computed multiplications is $2n^2/2 + n(n-1)/2 + n + n^3/2 = n(n^2 + 3n + 1)/2$.
\end{proof}

\begin{theorem}\label{commutative_recursive_2}
Let $R$ be a ring satisfying \eqref{involution}, let $n \geq 3$ be odd. Furthermore let $A$ and $B$ be $n \times n$ matrices over $R$. Then there is a non-bilinear algorithm using $n(n^2+3n+2)/2$ multiplications for computing the product $AB$.
\end{theorem}
\begin{proof}
Let $A$ and $B$ be matrices as in the Theorem. Let $A_{ij}$ denote the entries of $A$ and let $B_{ij}$ denote the entries of $B$ and let $C_{ij}$ denote the entries of $AB$. The product $AB$ can be computed as follows: For $i=1,3,5,\dots,n-2$ and $j=1,\dots,n$ let
\begin{align*}
C_{ij}&= \sum_{k=1}^n (A_{ik} + B_{jk}^T)(A_{(i+1)j}^T + B_{kj}) - (\sum_{k=1}^n A_{ik} + \sum_{k=1}^n B_{jk}^T)A_{(i+1)j}^T - \sum_{k=1}^n B_{jk}^TB_{kj}\\
C_{(i+1)j}&= \sum_{k=1}^n (A_{(i+1)k} + B_{jk}^T)(A_{ij}^T + B_{kj}) - (\sum_{k=1}^n A_{(i+1)k} + \sum_{k=1}^n B_{jk}^T)A_{ij}^T - \sum_{k=1}^n B_{jk}^TB_{kj}\\
C_{nj}&= \sum_{k=1}^n (A_{nk} + B_{jk}^T)(A_{nj}^T + B_{kj}) - (\sum_{k=1}^n A_{nk} + \sum_{k=1}^n B_{jk}^T)A_{nj}^T - \sum_{k=1}^n B_{jk}^TB_{kj}
\end{align*}

It is easy to verify these terms. The number of computed multiplications can be calculated analogously to Theorem \ref{commutative_recursive_1}.
\end{proof}

Let us compare the results of Theorems \ref{commutative_recursive_1} and \ref{commutative_recursive_2} with an overview of bilinear algorithms:
\begin{displaymath}
\begin{tabular}{|c|c|c|}
\hline
$n$ & Theorems \ref{commutative_recursive_1} and \ref{commutative_recursive_2} & \cite{ubersicht}\\
\hline
4 & 58 & 49\\
\hline
5 & 105 & 98\\
\hline
6 & 165 & 160\\
\hline
7 & 252 & 250\\
\hline
8 & 356 & 343\\
\hline
9 & \textbf{495} & 514\\
\hline
10 & \textbf{655} & 686\\
\hline
11 & \textbf{858} & 919\\
\hline
12 & 1086 & 1040\\
\hline
13 & \textbf{1365} & 1443\\
\hline
14 & \textbf{1673} & 1720\\
\hline
15 & \textbf{2040} & 2088\\
\hline
16 & 2440 & 2401\\
\hline
17 & \textbf{2907} & 2960\\
\hline
18 & 3411 & 3200\\
\hline
19 & \textbf{3990} & 4065\\
\hline
20 & 4610 & 4340\\
\hline
\end{tabular}
\end{displaymath}

\subsection{Optimization}
Note that for $k = j$ multiplications of the form $B_{jk}^TB_{kj}$ and $(A_{nk} + B_{jk}^T)(A_{nj}^T + B_{kj})$ can be computed more efficiently. Let
\begin{align*}
A = \begin{bmatrix}
A_{11} & A_{12}\\
A_{21} & A_{22}
\end{bmatrix}
\end{align*}

Let $C = A^TA$. Consider the computation of matrix $C$:
\begin{align*}
C_{11}&= A_{11}^TA_{11} + A_{21}^TA_{21}\\
C_{12}&= A_{11}^TA_{12} + A_{21}^TA_{22}\\
C_{21}&= A_{12}^TA_{11} + A_{22}^TA_{21}\\
C_{22}&= A_{12}^TA_{12} + A_{22}^TA_{22}
\end{align*}

Since we consider computations in a ring satisfying \eqref{involution} obviously $C_{12}^T = C_{21}$ holds. Thus, $A^TA$ can be computed using $4$ recursive calls and $2$ multiplications which we can compute with Theorems \ref{commutative_recursive_1} and \ref{commutative_recursive_2} again. Now, suppose we are working with matrices over a commutative ring in which $\sqrt{-1}$ exists for instance in the complex numbers. Then $A^TA$ can be computed using only $3$ recursive calls and $2$ multiplications likewise. To obtain this algorithm we adapted Strassen algorithm \cite{strassen:volker}:
\begin{theorem}
Let $R$ be a ring satisfying \eqref{involution}. Let $A$ be a $2 \times 2$ matrix over $R$ and let $R'$ be the underlying ring of the elements of $A$ and let there be an element $x \in R'$ such that $x^2 = -1$. Then the product $A^TA$ can be computed using $5$ multiplications.
\end{theorem}
\begin{proof}
Let
\begin{displaymath}
A = \begin{bmatrix}
A_{11} & A_{12}\\
A_{21} & A_{22}
\end{bmatrix}
\end{displaymath}
\begin{algorithm}\label{ATA_5}
Input: Matrix $A$
\begin{align*}
P_1&:=(A_{11}^T + \sqrt{-1}A_{22}^T)(A_{11} + \sqrt{-1}A_{22})\\
P_2&:=(A_{21}^T + A_{22}^T)(A_{21} + A_{22})\\
P_3&:=(A_{11}^T + A_{12}^T)(A_{11} + A_{12})\\
P_4&:=(A_{21}^T + \sqrt{-1}A_{11}^T)A_{22}\\
P_5&:=A_{11}^T(A_{12} - \sqrt{-1}A_{22})
\end{align*}
Output:
\begin{align*}
C_{11}&=P_1 + P_2 - P_4 - P_4^T\\
C_{12}&=P_4 + P_5=C_{21}^T\\
C_{22}&=-P_1 + P_3 - P_5 - P_5^T
\end{align*}
\end{algorithm}
\end{proof}

Note that the products $P_1$, $P_2$ and $P_3$ are recursive calls. Algorithm \ref{ATA_5} is even optimal in the sense that there is no algorithm for computing $A^TA$ that uses $5$ multiplications overall from which more than $3$ multiplications are recursive calls as the following result shows:
\begin{theorem}
Let $R$ be a ring satisfying \eqref{involution}. Let $A$ be a $2 \times 2$ matrix over $R$. Let $\gamma$ be an algorithm which requires $5$ multiplications to compute $A^TA$. Let $M_1, \dots, M_5$ be the multiplications computed by $\gamma$. Let $\alpha_i$ the left factor of $M_i$ and let $\beta_i$ the right factor of $M_i$. Then there are at most $3$ pairwise different indices $1 \leq i_1 < i_2 < i_3 \leq 5$ such that for $j=1,2,3$ the following holds:
\begin{displaymath}
\alpha_{i_j}^T = \beta_{i_j}
\end{displaymath}
\end{theorem}
\begin{proof}
Let $A$ be a matrix as in the Theorem. Let $\gamma$ be an algorithm, which requires $5$ multiplications for the computation of $A^TA$. Let $M_1, \dots, M_5$ be the computed  multiplications of $\gamma$. Let $\alpha_i$ the left factor of $M_i$ and let $\beta_i$ the right factor of $M_i$. Suppose there are $4$ pairwise different indices $1 \leq i_1 < i_2 < i_3 < i_4 \leq 5$ such that for $j=1,2,3,4$ we have $\alpha_{i_j}^T = \beta_{i_j}$. Without loss of generality we choose $i_j = j$. Let
\begin{align*}
E = \begin{bmatrix}
E_{11} & E_{12}\\
E_{21} & E_{22}
\end{bmatrix}
&&
F = \begin{bmatrix}
F_{11} & F_{12}\\
F_{21} & F_{22}
\end{bmatrix}
\end{align*}

We choose a new matrix $G$ as follows:
\begin{align*}
G^T = \begin{bmatrix}
E_{11} & E_{12}\\
F_{11}^T & F_{21}^T\\
E_{21} & E_{22}\\
F_{12}^T & F_{22}^T
\end{bmatrix}
\end{align*}

Then the product $EF$ is computed by $G^TG$. We apply $\gamma$ to compute the product $G^TG$. Since $G^T$ is a $4 \times 2$ matrix the multiplications $M_1, \dots, M_5$ compute products of $2 \times 1$ and $1 \times 2$ submatrices. For $i=1,\dots,4$ we have:
\begin{align*}
M_i = \alpha_i\beta_i = \beta_i^T\beta_i = \begin{bmatrix}
\beta_{i1}^T\\
\beta_{i2}^T
\end{bmatrix}
\begin{bmatrix}
\beta_{i1} & \beta_{i2}
\end{bmatrix}
\end{align*}
and
\begin{align*}
M_5 = \alpha_5\beta_5 = \begin{bmatrix}
\alpha_{51}\\
\alpha_{52}
\end{bmatrix}
\begin{bmatrix}
\beta_{51} & \beta_{52}
\end{bmatrix}
\end{align*}

For the computation of the products $M_1, \dots, M_5$ we need $16$ multiplications. For $i=1,\dots,4$ let:
\begin{align*}
M_{i1}&:= \beta_{i1}^T\beta_{i1}\\
M_{i2}&:= \beta_{i1}^T\beta_{i2}\\
M_{i3}&:= \beta_{i2}^T\beta_{i2}
\end{align*}

Furthermore let:
\begin{align*}
M_{51}&:=\alpha_{51}\beta_{51}\\
M_{52}&:=\alpha_{51}\beta_{52}\\
M_{53}&:=\alpha_{52}\beta_{51}\\
M_{54}&:=\alpha_{52}\beta_{52}
\end{align*}

For $i=1,\dots,4$ we have:
\begin{align*}
M_i = \begin{bmatrix}
M_{i1} & M_{i2}\\
M_{i2}^T & M_{i3}
\end{bmatrix}
\end{align*}

Besides:
\begin{align*}
M_5 = \begin{bmatrix}
M_{51} & M_{52}\\
M_{53} & M_{54}
\end{bmatrix}
\end{align*}

Let
\begin{align*}
G^TG = D = \begin{bmatrix}
D_1 & D_2\\
D_3 & D_4
\end{bmatrix}
\end{align*}

be the product computed by $\gamma$ in which for $i=1,\dots,4$ let
\begin{align*}
D_i = \begin{bmatrix}
D_{i1} & D_{i2}\\
D_{i3} & D_{i4}
\end{bmatrix}
\end{align*}

It follows:
\begin{align*}
D_1 = \begin{bmatrix}
E_{11} & E_{12}\\
F_{11}^T & F_{21}^T
\end{bmatrix}
\cdot
\begin{bmatrix}
E_{11}^T & F_{11}\\
E_{12}^T & F_{21}
\end{bmatrix}
&&
D_2 = \begin{bmatrix}
E_{11} & E_{12}\\
F_{11}^T & F_{21}^T
\end{bmatrix}
\cdot
\begin{bmatrix}
E_{21}^T & F_{12}\\
E_{22}^T & F_{22}
\end{bmatrix}
\\
D_3 = \begin{bmatrix}
E_{21} & E_{22}\\
F_{12}^T & F_{22}^T
\end{bmatrix}
\cdot
\begin{bmatrix}
E_{11}^T & F_{11}\\
E_{12}^T & F_{21}
\end{bmatrix}
&&
D_4 = \begin{bmatrix}
E_{21} & E_{22}\\
F_{12}^T & F_{22}^T
\end{bmatrix}
\cdot
\begin{bmatrix}
E_{21}^T & F_{12}\\
E_{22}^T & F_{22}
\end{bmatrix}
\end{align*}

For $k=1,\dots,5$ and $i=1,\dots,4$ there are constants $d_{ki}$ and $e_{ki}$ such that:
\begin{equation}\label{BerechnungDi}
D_i = \sum_{k=1}^5 d_{ki}M_k + \sum_{k=1}^5 e_{ki}M_k^T
\end{equation}

But we are only interested in the subproduct $EF$. We have:
\begin{align*}
EF = \begin{bmatrix}
E_{11} & E_{12}\\
E_{21} & E_{22}
\end{bmatrix}
\cdot
\begin{bmatrix}
F_{11} & F_{12}\\
F_{21} & F_{22}
\end{bmatrix}
=
\begin{bmatrix}
D_{12} & D_{22}\\
D_{32} & D_{42}
\end{bmatrix}
\end{align*}

However, by equation \eqref{BerechnungDi} only the products $M_{12}, M_{22}, M_{32}, M_{42}, M_{52}$ and $M_{53}$ come into question for computing the product $EF$. By this fact we have found an algorithm to compute the product of two $2 \times 2$ matrices using only $6$ multiplications. However, it was shown that even when using commutativity at least $7$ multiplications are required to compute the product of two $2 \times 2$ matrices \cite{hopcroft:kerr:min, winograd}. This is a contradiction.
\end{proof}

\subsection{Approximate algorithms}
Approximate algorithms for matrix multiplication were first introduced by Bini et al. \cite{bini}. The idea of approximate algorithms for matrix multiplication is to introduce an additional factor usually denoted by $\varepsilon$ or $\lambda$ that allows to compute an error term in addition to the correct product. Letting $\varepsilon$ tend to $0$ allows to compute the product with any accuracy. In addition approximate algorithms can be used to obtain an exact result. The interested reader is referred to \cite{bini_approx_exact}.

From the Theorems \ref{commutative_recursive_1} and \ref{commutative_recursive_2} we can easily obtain approximate algorithms.
\begin{theorem}\label{commutative_recursive_1_approximate}
Let $R$ be a ring satisfying \eqref{involution}, let $n \geq 2$ be even. Furthermore let $A$ and $B$ be $n \times n$ matrices over $R$. Then there is an approximate non-bilinear algorithm using $n(n^2+2n)/2$ multiplications for computing the product $AB$.
\end{theorem}
\begin{proof}
Let $A$ and $B$ be matrices as in the Theorem. Let $A_{ij}$ denote the entries of $A$ and let $B_{ij}$ denote the entries of $B$ and let $C_{ij}$ denote the entries of $AB$. The product $AB$ can be computed as follows: For $i=1,3,5,\dots,n-1$ and $j=1,\dots,n$ let
\begin{align*}
C_{ij}&= \frac{1}{\varepsilon}\left(\sum_{k=1}^n (A_{ik} + B_{jk}^T\varepsilon)(A_{(i+1)j}^T + B_{kj}\varepsilon) - (\sum_{k=1}^n A_{ik} + \sum_{k=1}^n B_{jk}^T\varepsilon)A_{(i+1)j}^T\right)\\
C_{(i+1)j}&= \frac{1}{\varepsilon}\left(\sum_{k=1}^n (A_{(i+1)k} + B_{jk}^T\varepsilon)(A_{ij}^T + B_{kj}\varepsilon) - (\sum_{k=1}^n A_{(i+1)k} + \sum_{k=1}^n B_{jk}^T\varepsilon)A_{ij}^T\right)
\end{align*}
\end{proof}
\begin{theorem}\label{commutative_recursive_2_approximate}
Let $R$ be a ring satisfying \eqref{involution}, let $n \geq 3$ be odd. Furthermore let $A$ and $B$ be $n \times n$ matrices over $R$. Then there is an approximate non-bilinear algorithm using $n(n^2+2n+1)/2$ multiplications for computing the product $AB$.
\end{theorem}
\begin{proof}
Let $A$ and $B$ be matrices as in the Theorem. Let $A_{ij}$ denote the entries of $A$ and let $B_{ij}$ denote the entries of $B$ and let $C_{ij}$ denote the entries of $AB$. The product $AB$ can be computed as follows: For $i=1,3,5,\dots,n-2$ and $j=1,\dots,n$ let
\begin{align*}
C_{ij}&= \frac{1}{\varepsilon}\left(\sum_{k=1}^n (A_{ik} + B_{jk}^T\varepsilon)(A_{(i+1)j}^T + B_{kj}\varepsilon) - (\sum_{k=1}^n A_{ik} + \sum_{k=1}^n B_{jk}^T\varepsilon)A_{(i+1)j}^T\right)\\
C_{(i+1)j}&= \frac{1}{\varepsilon}\left(\sum_{k=1}^n (A_{(i+1)k} + B_{jk}^T\varepsilon)(A_{ij}^T + B_{kj}\varepsilon) - (\sum_{k=1}^n A_{(i+1)k} + \sum_{k=1}^n B_{jk}^T\varepsilon)A_{ij}^T\right)\\
C_{nj}&= \frac{1}{\varepsilon}\left(\sum_{k=1}^n (A_{nk} + B_{jk}^T\varepsilon)(A_{nj}^T + B_{kj}\varepsilon) - (\sum_{k=1}^n A_{nk} + \sum_{k=1}^n B_{jk}^T\varepsilon)A_{nj}^T\right)
\end{align*}
\end{proof}

Note that the algorithms of Theorems \ref{commutative_recursive_1_approximate} and \ref{commutative_recursive_2_approximate} are of degree $1$. In the case where $n = 5$ Theorem \ref{commutative_recursive_2_approximate} yields an algorithm using only $90$ multiplications. The best approximate bilinear algorithm for $5 \times 5$ matrices we are aware of was published by Smirnov \cite{smirnov5x5} and uses also $90$ multiplications. But in this special case we were even able to find an approximate non-bilinear algorithm that uses only $89$ multiplications. Nevertheless this algorithm is only of degree $2$.
\begin{theorem}
There is an approximate non-bilinear algorithm that computes the product of two $5 \times 5$ matrices with only $89$ multiplications.
\end{theorem}
\begin{proof}
Let
\begin{align*}
A&= \begin{bmatrix}
A_{11} & A_{12} & A_{13} & A_{14} & A_{15}\\
A_{21} & A_{22} & A_{23} & A_{24} & A_{25}\\
A_{31} & A_{32} & A_{33} & A_{34} & A_{35}\\
A_{41} & A_{42} & A_{43} & A_{44} & A_{45}\\
A_{51} & A_{52} & A_{53} & A_{54} & A_{55}
\end{bmatrix}
&&
B = \begin{bmatrix}
B_{11} & B_{12} & B_{13} & B_{14} & B_{15}\\
B_{21} & B_{22} & B_{23} & B_{24} & B_{25}\\
B_{31} & B_{32} & B_{33} & B_{34} & B_{35}\\
B_{41} & B_{42} & B_{43} & B_{44} & B_{45}\\
B_{51} & B_{52} & B_{53} & B_{54} & B_{55}
\end{bmatrix}
\end{align*}
Let
\begin{displaymath}
AB = \begin{bmatrix}
C_{11} & C_{12} & C_{13} & C_{14} & C_{15}\\
C_{21} & C_{22} & C_{23} & C_{24} & C_{25}\\
C_{31} & C_{32} & C_{33} & C_{34} & C_{35}\\
C_{41} & C_{42} & C_{43} & C_{44} & C_{45}\\
C_{51} & C_{52} & C_{53} & C_{54} & C_{55}
\end{bmatrix}
\end{displaymath}
\begin{algorithm}\label{5x5_89}
Input: Matrices $A$ and $B$. Let:
\begin{align*}
M_1&:=(A_{11} + B_{11}^T\varepsilon^2)(B_{11}\varepsilon^2 + A_{21}^T) && M_6:=(A_{11} + B_{21}^T\varepsilon^2)(B_{12}\varepsilon^2 + A_{22}^T)\\
M_2&:=(A_{12} + B_{12}^T\varepsilon^2)(B_{21}\varepsilon^2 + A_{21}^T) && M_7:=(A_{12} + B_{22}^T\varepsilon^2)(B_{22}\varepsilon^2 + A_{22}^T)\\
M_3&:=(A_{13} + B_{13}^T\varepsilon^2)(B_{31}\varepsilon^2 + A_{21}^T) && M_8:=(A_{13} + B_{23}^T\varepsilon^2)(B_{32}\varepsilon^2 + A_{22}^T)\\
M_4&:=(A_{14} + B_{14}^T\varepsilon^2)(B_{41}\varepsilon^2 + A_{21}^T) && M_9:=(A_{14} + B_{24}^T\varepsilon^2)(B_{42}\varepsilon^2 + A_{22}^T)\\
M_5&:=(A_{15} + B_{15}^T\varepsilon^2)(B_{51}\varepsilon^2 + A_{21}^T) && M_{10}:=(A_{15} + B_{25}^T\varepsilon^2)(B_{52}\varepsilon^2 + A_{22}^T)
\end{align*}
\begin{align*}
M_{11}&:=(A_{11} + B_{31}^T\varepsilon^2)(B_{13}\varepsilon^2 + A_{23}^T) && M_{16}:=(A_{11} + B_{41}^T\varepsilon^2)(B_{14}\varepsilon^2 + A_{24}^T)\\
M_{12}&:=(A_{12} + B_{32}^T\varepsilon^2)(B_{23}\varepsilon^2 + A_{23}^T) && M_{17}:=(A_{12} + B_{42}^T\varepsilon^2)(B_{24}\varepsilon^2 + A_{24}^T)\\
M_{13}&:=(A_{13} + B_{33}^T\varepsilon^2)(B_{33}\varepsilon^2 + A_{23}^T) && M_{18}:=(A_{13} + B_{43}^T\varepsilon^2)(B_{34}\varepsilon^2 + A_{24}^T)\\
M_{14}&:=(A_{14} + B_{34}^T\varepsilon^2)(B_{43}\varepsilon^2 + A_{23}^T) && M_{19}:=(A_{14} + B_{44}^T\varepsilon^2)(B_{44}\varepsilon^2 + A_{24}^T)\\
M_{15}&:=(A_{15} + B_{35}^T\varepsilon^2)(B_{53}\varepsilon^2 + A_{23}^T) && M_{20}:=(A_{15} + B_{45}^T\varepsilon^2)(B_{54}\varepsilon^2 + A_{24}^T)\\
\end{align*}
\begin{align*}
M_{21}&:=(A_{11} + A_{12} + A_{13} + A_{14} + A_{15} + B_{11}^T\varepsilon^2 + B_{12}^T\varepsilon^2 + B_{13}^T\varepsilon^2 + B_{14}^T\varepsilon^2 + B_{15}^T\varepsilon^2)A_{21}^T\\
M_{22}&:=(A_{11} + A_{12} + A_{13} + A_{14} + A_{15} + B_{21}^T\varepsilon^2 + B_{22}^T\varepsilon^2 + B_{23}^T\varepsilon^2 + B_{24}^T\varepsilon^2 + B_{25}^T\varepsilon^2)A_{22}^T\\
M_{23}&:=(A_{11} + A_{12} + A_{13} + A_{14} + A_{15} + B_{31}^T\varepsilon^2 + B_{32}^T\varepsilon^2 + B_{33}^T\varepsilon^2 + B_{34}^T\varepsilon^2 + B_{35}^T\varepsilon^2)A_{23}^T\\
M_{24}&:=(A_{11} + A_{12} + A_{13} + A_{14} + A_{15} + B_{41}^T\varepsilon^2 + B_{42}^T\varepsilon^2 + B_{43}^T\varepsilon^2 + B_{44}^T\varepsilon^2 + B_{45}^T\varepsilon^2)A_{24}^T
\end{align*}
\begin{align*}
M_{25}&:=(A_{21} + A_{22} + A_{23} + A_{24} + B_{11}^T\varepsilon^2 + B_{12}^T\varepsilon^2 + B_{13}^T\varepsilon^2 + B_{14}^T\varepsilon^2 - B_{15}^T\varepsilon^3)A_{11}^T\\
M_{26}&:=(A_{21} + A_{22} + A_{23} + A_{24} + B_{21}^T\varepsilon^2 + B_{22}^T\varepsilon^2 + B_{23}^T\varepsilon^2 + B_{24}^T\varepsilon^2 - B_{25}^T\varepsilon^3)A_{12}^T\\
M_{27}&:=(A_{21} + A_{22} + A_{23} + A_{24} + B_{31}^T\varepsilon^2 + B_{32}^T\varepsilon^2 + B_{33}^T\varepsilon^2 + B_{34}^T\varepsilon^2 - B_{35}^T\varepsilon^3)A_{13}^T\\
M_{28}&:=(A_{21} + A_{22} + A_{23} + A_{24} + B_{41}^T\varepsilon^2 + B_{42}^T\varepsilon^2 + B_{43}^T\varepsilon^2 + B_{44}^T\varepsilon^2 - B_{45}^T\varepsilon^3)A_{14}^T\\
M_{29}&:=(A_{21} + A_{22} + A_{23} + A_{24} + B_{51}^T\varepsilon^2 + B_{52}^T\varepsilon^2 + B_{53}^T\varepsilon^2 + B_{54}^T\varepsilon^2 - B_{55}^T\varepsilon^3)A_{15}^T
\end{align*}
\begin{align*}
M_{30}&:=(A_{11} + A_{12} + A_{13} + A_{14} + A_{15})(A_{21}^T + A_{22}^T + A_{23}^T + A_{24}^T)
\end{align*}
\begin{align*}
M_{31}&:=(A_{31} + B_{11}^T\varepsilon^2)(B_{11}\varepsilon^2 + A_{41}^T) && M_{36}:=(A_{31} + B_{21}^T\varepsilon^2)(B_{12}\varepsilon^2 + A_{42}^T)\\
M_{32}&:=(A_{32} + B_{12}^T\varepsilon^2)(B_{21}\varepsilon^2 + A_{41}^T) && M_{37}:=(A_{32} + B_{22}^T\varepsilon^2)(B_{22}\varepsilon^2 + A_{42}^T)\\
M_{33}&:=(A_{33} + B_{13}^T\varepsilon^2)(B_{31}\varepsilon^2 + A_{41}^T) && M_{38}:=(A_{33} + B_{23}^T\varepsilon^2)(B_{32}\varepsilon^2 + A_{42}^T)\\
M_{34}&:=(A_{34} + B_{14}^T\varepsilon^2)(B_{41}\varepsilon^2 + A_{41}^T) && M_{39}:=(A_{34} + B_{24}^T\varepsilon^2)(B_{42}\varepsilon^2 + A_{42}^T)\\
M_{35}&:=(A_{35} + B_{15}^T\varepsilon^2)(B_{51}\varepsilon^2 + A_{41}^T) && M_{40}:=(A_{35} + B_{25}^T\varepsilon^2)(B_{52}\varepsilon^2 + A_{42}^T)
\end{align*}
\begin{align*}
M_{41}&:=(A_{31} + B_{31}^T\varepsilon^2)(B_{13}\varepsilon^2 + A_{43}^T) && M_{46}:=(A_{31} + B_{41}^T\varepsilon^2)(B_{14}\varepsilon^2 + A_{44}^T)\\
M_{42}&:=(A_{32} + B_{32}^T\varepsilon^2)(B_{23}\varepsilon^2 + A_{43}^T) && M_{47}:=(A_{32} + B_{42}^T\varepsilon^2)(B_{24}\varepsilon^2 + A_{44}^T)\\
M_{43}&:=(A_{33} + B_{33}^T\varepsilon^2)(B_{33}\varepsilon^2 + A_{43}^T) && M_{48}:=(A_{33} + B_{43}^T\varepsilon^2)(B_{34}\varepsilon^2 + A_{44}^T)\\
M_{44}&:=(A_{34} + B_{34}^T\varepsilon^2)(B_{43}\varepsilon^2 + A_{43}^T) && M_{49}:=(A_{34} + B_{44}^T\varepsilon^2)(B_{44}\varepsilon^2 + A_{44}^T)\\
M_{45}&:=(A_{35} + B_{35}^T\varepsilon^2)(B_{53}\varepsilon^2 + A_{43}^T) && M_{50}:=(A_{35} + B_{45}^T\varepsilon^2)(B_{54}\varepsilon^2 + A_{44}^T)
\end{align*}
\begin{align*}
M_{51}&:=(A_{31} + A_{32} + A_{33} + A_{34} + A_{35} + B_{11}^T\varepsilon^2 + B_{12}^T\varepsilon^2 + B_{13}^T\varepsilon^2 + B_{14}^T\varepsilon^2 + B_{15}^T\varepsilon^2)A_{41}^T\\
M_{52}&:=(A_{31} + A_{32} + A_{33} + A_{34} + A_{35} + B_{21}^T\varepsilon^2 + B_{22}^T\varepsilon^2 + B_{23}^T\varepsilon^2 + B_{24}^T\varepsilon^2 + B_{25}^T\varepsilon^2)A_{42}^T\\
M_{53}&:=(A_{31} + A_{32} + A_{33} + A_{34} + A_{35} + B_{31}^T\varepsilon^2 + B_{32}^T\varepsilon^2 + B_{33}^T\varepsilon^2 + B_{34}^T\varepsilon^2 + B_{35}^T\varepsilon^2)A_{43}^T\\
M_{54}&:=(A_{31} + A_{32} + A_{33} + A_{34} + A_{35} + B_{41}^T\varepsilon^2 + B_{42}^T\varepsilon^2 + B_{43}^T\varepsilon^2 + B_{44}^T\varepsilon^2 + B_{45}^T\varepsilon^2)A_{44}^T
\end{align*}
\begin{align*}
M_{55}&:=(A_{41} + A_{42} + A_{43} + A_{44} + B_{11}^T\varepsilon^2 + B_{12}^T\varepsilon^2 + B_{13}^T\varepsilon^2 + B_{14}^T\varepsilon^2 - B_{15}^T\varepsilon^3)A_{31}^T\\
M_{56}&:=(A_{41} + A_{42} + A_{43} + A_{44} + B_{21}^T\varepsilon^2 + B_{22}^T\varepsilon^2 + B_{23}^T\varepsilon^2 + B_{24}^T\varepsilon^2 - B_{25}^T\varepsilon^3)A_{32}^T\\
M_{57}&:=(A_{41} + A_{42} + A_{43} + A_{44} + B_{31}^T\varepsilon^2 + B_{32}^T\varepsilon^2 + B_{33}^T\varepsilon^2 + B_{34}^T\varepsilon^2 - B_{35}^T\varepsilon^3)A_{33}^T\\
M_{58}&:=(A_{41} + A_{42} + A_{43} + A_{44} + B_{41}^T\varepsilon^2 + B_{42}^T\varepsilon^2 + B_{43}^T\varepsilon^2 + B_{44}^T\varepsilon^2 - B_{45}^T\varepsilon^3)A_{34}^T\\
M_{59}&:=(A_{41} + A_{42} + A_{43} + A_{44} + B_{51}^T\varepsilon^2 + B_{52}^T\varepsilon^2 + B_{53}^T\varepsilon^2 + B_{54}^T\varepsilon^2 - B_{55}^T\varepsilon^3)A_{35}^T
\end{align*}
\begin{align*}
M_{60}&:=(A_{31} + A_{32} + A_{33} + A_{34} + A_{35})(A_{41}^T + A_{42}^T + A_{43}^T + A_{44}^T)
\end{align*}
\begin{align*}
M_{61}&:=(A_{25} + A_{51}\varepsilon + A_{45})(B_{51} + B_{11}\varepsilon) && M_{66}:=(A_{51}\varepsilon + A_{45})B_{51}\\
M_{62}&:=(A_{25} + A_{52}\varepsilon + A_{45})(B_{52} + B_{21}\varepsilon) && M_{67}:=(A_{52}\varepsilon + A_{45})B_{52}\\
M_{63}&:=(A_{25} + A_{53}\varepsilon + A_{45})(B_{53} + B_{31}\varepsilon) && M_{68}:=(A_{53}\varepsilon + A_{45})B_{53}\\
M_{64}&:=(A_{25} + A_{54}\varepsilon + A_{45})(B_{54} + B_{41}\varepsilon) && M_{69}:=(A_{54}\varepsilon + A_{45})B_{54}\\
M_{65}&:=(A_{25} + A_{55}\varepsilon + A_{45})(B_{55} + B_{51}\varepsilon) && M_{70}:=(A_{55}\varepsilon + A_{45})B_{55}
\end{align*}
\begin{align*}
M_{71}&:=(A_{25} + A_{45})(B_{51} + B_{52} + B_{53} + B_{54} + B_{55} + B_{11}\varepsilon + B_{21}\varepsilon + B_{31}\varepsilon + B_{41}\varepsilon + B_{51}\varepsilon)
\end{align*}
\begin{align*}
M_{72}&:=A_{45}(B_{51} + B_{52} + B_{53} + B_{54} + B_{55})
\end{align*}
\begin{align*}
M_{73}&:=(A_{51} + B_{42}^T\varepsilon^2)(B_{12}\varepsilon^2 + A_{54}^T) && M_{76}:=(A_{51} + B_{52}^T\varepsilon^2)(B_{13}\varepsilon^2 + A_{55}^T)\\
M_{74}&:=(A_{52} + B_{43}^T\varepsilon^2)(B_{22}\varepsilon^2 + A_{54}^T) && M_{77}:=(A_{52} + B_{53}^T\varepsilon^2)(B_{23}\varepsilon^2 + A_{55}^T)\\
M_{75}&:=(A_{53} + B_{44}^T\varepsilon^2)(B_{32}\varepsilon^2 + A_{54}^T) && M_{78}:=(A_{53} + B_{54}^T\varepsilon^2)(B_{33}\varepsilon^2 + A_{55}^T)
\end{align*}
\begin{align*}
M_{79}&:=(A_{51} + A_{52} + A_{53} + B_{42}^T\varepsilon^2 + B_{43}^T\varepsilon^2 + B_{44}^T\varepsilon^2)A_{54}^T\\
M_{80}&:=(A_{51} + A_{52} + A_{53} + B_{52}^T\varepsilon^2 + B_{53}^T\varepsilon^2 + B_{54}^T\varepsilon^2)A_{55}^T
\end{align*}
\begin{align*}
M_{81}&:=(A_{54} + A_{55} + B_{12}^T\varepsilon^2 + B_{13}^T\varepsilon^2 - B_{14}^T\varepsilon^3)A_{51}^T\\
M_{82}&:=(A_{54} + A_{55} + B_{22}^T\varepsilon^2 + B_{23}^T\varepsilon^2 - B_{24}^T\varepsilon^3)A_{52}^T\\
M_{83}&:=(A_{54} + A_{55} + B_{32}^T\varepsilon^2 + B_{33}^T\varepsilon^2 - B_{34}^T\varepsilon^3)A_{53}^T
\end{align*}
\begin{align*}
M_{84}&:=(A_{51} + A_{52} + A_{53})(A_{54}^T + A_{55}^T)
\end{align*}
\begin{align*}
M_{85}&:=A_{51}B_{15}\\
M_{86}&:=A_{52}B_{25}\\
M_{87}&:=A_{53}B_{35}\\
M_{88}&:=A_{54}B_{45}\\
M_{89}&:=A_{55}B_{55}
\end{align*}
Output:
\begin{align*}
C_{11}&=\frac{1}{\varepsilon^2}(M_1 + M_2 + M_3 + M_4 + M_5 - M_{21})\\
C_{12}&=\frac{1}{\varepsilon^2}(M_6 + M_7 + M_8 + M_9 + M_{10} - M_{22})\\
C_{13}&=\frac{1}{\varepsilon^2}(M_{11} + M_{12} + M_{13} + M_{14} + M_{15} - M_{23})\\
C_{14}&=\frac{1}{\varepsilon^2}(M_{16} + M_{17} + M_{18} + M_{19} + M_{20} - M_{24})\\
C_{15}&=\frac{1}{\varepsilon^3}(M_1 + M_2 + M_3 + M_4 + M_5 + M_6 + M_7 + M_8 + M_9 + M_{10} + M_{11} + M_{12}\\
& + M_{13} + M_{14} + M_{15} + M_{16} + M_{17} + M_{18} + M_{19} + M_{20} - M_{21} - M_{22} - M_{23} - M_{24}\\
& - M_{25}^T - M_{26}^T - M_{27}^T - M_{28}^T - M_{29}^T + M_{30})
\end{align*}
\begin{align*}
C_{21}&=\frac{1}{\varepsilon^2}(M_1^T + M_6^T + M_{11}^T + M_{16}^T - M_{25}) + M_{61} - M_{66}\\
C_{22}&=\frac{1}{\varepsilon^2}(M_2^T + M_7^T + M_{12}^T + M_{17}^T - M_{26}) + M_{62} - M_{67}\\
C_{23}&=\frac{1}{\varepsilon^2}(M_3^T + M_8^T + M_{13}^T + M_{18}^T - M_{27}) + M_{63} - M_{68}\\
C_{24}&=\frac{1}{\varepsilon^2}(M_4^T + M_9^T + M_{14}^T + M_{19}^T - M_{28}) + M_{64} - M_{69}\\
C_{25}&=\frac{1}{\varepsilon^2}(M_5^T + M_{10}^T + M_{15}^T + M_{20}^T - M_{29}) + M_{65} - M_{70}
\end{align*}
\begin{align*}
C_{31}&=\frac{1}{\varepsilon^2}(M_{31} + M_{32} + M_{33} + M_{34} + M_{35} - M_{51})\\
C_{32}&=\frac{1}{\varepsilon^2}(M_{36} + M_{37} + M_{38} + M_{39} + M_{40} - M_{52})\\
C_{33}&=\frac{1}{\varepsilon^2}(M_{41} + M_{42} + M_{43} + M_{44} + M_{45} - M_{53})\\
C_{34}&=\frac{1}{\varepsilon^2}(M_{46} + M_{47} + M_{48} + M_{49} + M_{50} - M_{54})\\
C_{35}&=\frac{1}{\varepsilon^3}(M_{31} + M_{32} + M_{33} + M_{34} + M_{35} + M_{36} + M_{37} + M_{38} + M_{39} + M_{40} + M_{41} + M_{42}\\
& + M_{43} + M_{44} + M_{45} + M_{46} + M_{47} + M_{48} + M_{49} + M_{50} - M_{51} - M_{52} - M_{53} - M_{54}\\
& - M_{55}^T - M_{56}^T - M_{57}^T - M_{58}^T - M_{59}^T + M_{60})
\end{align*}
\begin{align*}
C_{41}&=\frac{1}{\varepsilon^2}(M_{31}^T + M_{36}^T + M_{41}^T + M_{46}^T - M_{55}) + M_{66}\\
C_{42}&=\frac{1}{\varepsilon^2}(M_{32}^T + M_{37}^T + M_{42}^T + M_{47}^T - M_{56}) + M_{67}\\
C_{43}&=\frac{1}{\varepsilon^2}(M_{33}^T + M_{38}^T + M_{43}^T + M_{48}^T - M_{57}) + M_{68}\\
C_{44}&=\frac{1}{\varepsilon^2}(M_{34}^T + M_{39}^T + M_{44}^T + M_{49}^T - M_{58}) + M_{69}\\
C_{45}&=\frac{1}{\varepsilon^2}(M_{35}^T + M_{40}^T + M_{45}^T + M_{50}^T - M_{59}) + M_{70}
\end{align*}
\begin{align*}
C_{51}&=\frac{1}{\varepsilon^2}(M_{61} + M_{62} + M_{63} + M_{64} + M_{65} - M_{66} - M_{67} - M_{68} - M_{69} - M_{70} - M_{71} + M_{72})\\
C_{52}&=\frac{1}{\varepsilon^2}(M_{73} + M_{74} + M_{75} - M_{79} + M_{73}^T + M_{76}^T - M_{81})\\
C_{53}&=\frac{1}{\varepsilon^2}(M_{76} + M_{77} + M_{78} - M_{80} + M_{74}^T + M_{77}^T - M_{82})\\
C_{54}&=\frac{1}{\varepsilon^3}(M_{73} + M_{74} + M_{75} + M_{76} + M_{77} + M_{78} - M_{79} - M_{80} - M_{81}^T - M_{82}^T - M_{83}^T + M_{84} + M_{75}^T\varepsilon + M_{78}^T\varepsilon - M_{83}\varepsilon)\\
C_{55}&= M_{85} + M_{86} + M_{87} + M_{88} + M_{89}
\end{align*}
\end{algorithm}
\end{proof}

\section{Open Problem}
We denote the minimal number of multiplications required by a non-bilinear algorithm to compute the product of two $n \times n$ matrices by non-bilinear rank and use $R_{nb}(\langle n,n,n \rangle)$ for that. Since by definition every bilinear algorithm is a non-bilinear algorithm it is clear that
\begin{displaymath}
R_{nb}(\langle n,n,n \rangle) \leq R(\langle n,n,n \rangle)
\end{displaymath}
But it remains open if equality holds.

\section{Acknowledgment}
I am grateful to Michael Figelius and Markus Lohrey for helpful comments.

\bibliographystyle{abbrv}

\end{document}